%% file: main.tex
\documentclass{IEEEtran4PSCC}
\ifCLASSINFOpdf
   \usepackage[pdftex]{graphicx}
\else
   \usepackage[dvips]{graphicx}
\fi
\usepackage[cmex10]{amsmath}
\graphicspath{ {./images/} }
\usepackage[utf8]{inputenc} 
\usepackage[T1]{fontenc}    
\usepackage{hyperref}       
\usepackage{url}            
\usepackage{booktabs}       
\usepackage{amsfonts}       
\usepackage{nicefrac}       
\usepackage{microtype}      
\usepackage{xcolor}         
\usepackage{amsthm}
\newtheorem{theorem}{Theorem}

\newtheorem{assumption}{Assumption}
\newtheorem{lemma}[theorem]{Lemma}
\newtheorem{definition}{Definition}
\usepackage{mathtools}
\usepackage{semantic}
\usepackage{algorithm}
\usepackage{algpseudocode}
\usepackage{enumitem}
\usepackage{amssymb}

\newcommand{\x}{\mathbf{x}}
\newcommand{\p}{\mathbf{p}}
\newcommand{\e}{\mathbf{e}}
\newcommand{\xis}{\x_i^{*}}
\newcommand{\T}{\intercal}

\newcommand{\edit}[1]{#1}

\DeclareMathOperator{\LSE}{LSE}

\makeatletter
\let\old@ps@headings\ps@headings
\let\old@ps@IEEEtitlepagestyle\ps@IEEEtitlepagestyle
\def\psccfooter#1{%
    \def\ps@headings{%
        \old@ps@headings%
        \def\@oddfoot{\strut\hfill#1\hfill\strut}%
        \def\@evenfoot{\strut\hfill#1\hfill\strut}%
    }%
    \def\ps@IEEEtitlepagestyle{%
        \old@ps@IEEEtitlepagestyle%
        \def\@oddfoot{\strut\hfill#1\hfill\strut}%
        \def\@evenfoot{\strut\hfill#1\hfill\strut}%
    }%
    \ps@headings%
}
\makeatother

\psccfooter{%
        \parbox{\textwidth}{\hrulefill \\ \small{23rd Power Systems Computation Conference} \hfill \begin{minipage}{0.2\textwidth}\centering \vspace*{4pt} \includegraphics[scale=0.06]{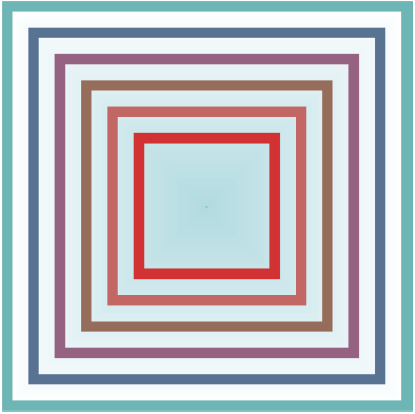}\\\small{PSCC 2024} \end{minipage} \hfill \small{Paris, France --- June 4 -- 7, 2024}}%
}

\begin{document}
%
\title{Socially Optimal Energy Usage via Adaptive Pricing}

\author{
\IEEEauthorblockN{Jiayi Li, Matthew Motoki and Baosen Zhang \\}
\IEEEauthorblockA{Electrical and Computer Engineering, University of Washington \\
\{ljy9712, mmotoki, zhangbao\}@uw.edu}
}



\maketitle

\begin{abstract}
A central challenge in using price signals to coordinate the electricity consumption of a group of users is the operator's lack of knowledge of the users \edit{due to privacy concerns}.  In this paper, we develop a two\edit{-}time-scale incentive mechanism that alternately updates between the users and a system operator. As long as the users can optimize their own consumption subject to a given price, the operator does not need to know or attempt to learn any private information of the users \edit{for price design. 
Users adjust their consumption following the price and the system redesigns the price based on the users' consumption.} We show that under \edit{mild}
assumptions, this iterative process converges to the social welfare solution. In particular, the cost of the users need not always be convex and its consumption can be the output of a \edit{machine learning-based load control algorithm. }


\end{abstract}


\thanksto{The authors are partially supported by NSF award CNS-1931718 and the Washington Clean Energy Institute}

\section{Introduction} \label{sec:intro}

\input{intro.tex}

\section{Problem Formulation and Pricing Updates} \label{sec:model}
\input{model.tex}
\input{algorithm.tex}

\section{Theoretical Results} \label{sec:theory}
\input{theoretical_results.tex}

\section{Simulation Results} \label{sec:simulation}
\input{simulation_results.tex}

\section{Conclusion} \label{sec:conclusion}
In this paper, we consider the problem of using price signals to coordinate the electricity consumption of a group of users. We develop a two\edit{-}time-scale incentive mechanism that alternately updates between the users and a system operator. We assume that users can optimize their actions given a price, and no private information is exchanged between the users and the system operator. In turn, the system operator does not need to learn the cost of the users. We show that the iterative algorithm converges to the social\edit{ly} optimal solution.  

\edit{The algorithms proposed in this paper can be implemented by any smart device that can receive and respond to information from a system operator. For example, a home with a Nest thermostat could easily negotiate with a system to reach the price equilibrium without any manual intervention of a user. A key assumption in this paper is that each of the users are price takers, in the sense that they lack any market power to explicitly manipulate the prices. Thus an important future direction for us is to consider the price-anticipatory behavior of users.}

\appendices
\section{Proof of Theorem~\ref{thm:t_multi}} \label{app:t_multi}
\input{appendix.tex}

\bibliographystyle{IEEEtran} 
\bibliography{mybib.bib}

\end{document}

%% file: intro.tex
We study the coordination of the electricity usage of a group of users by an operator. This question has been studied extensively by the community, for example, in the context of demand response, customer aggregation, and virtual power plants~(see~\cite{pinson2014benefits,zhang2017robust,tindemans2015decentralized,sarker2015optimal} and the references within). The common setup is where each user is endowed with a cost (or utility) function, and the operator seeks to minimize a global cost function that is made up of the individual costs of the users and social welfare considerations. 

A central challenge in these problems is that the cost functions of the users may not be known to the operator. Furthermore, users themselves may not be able to provide an analytical description. For example, using machine learning to schedule and manage demand has been a very active area of research (see, e.g.~\cite{o2010residential,li2019constrained,vazquez2019reinforcement,wang2020deep,shi2020multi} and the references within). However, most learning algorithms minimize some cost or discomfort by changing \edit{the} behavior of the users or their devices and do not provide a cost function in term\edit{s} of power that can be easily optimized. It's important to note that we do not mean that a cost function does not exist, rather, it is often not revealed when learning algorithms are applied.  In many existing algorithms, however, the operator needs to somehow learn the cost functions of the users~\cite{Li19,khezeli2017risk,zheng2020incentive}. This restricts the users to having simple (typically quadratic) cost functions and is often too restrictive to be implemented in practice.

In this paper, we overcome this challenge by deploying a two\edit{-}time-scale incentive mechanism that alternatively updates between the operator and the users. More concretely, the actions of the users are characterized by their electricity consumption and the action of the operator is in selecting a price vector. This price is an externality to the users, and a user optimizes its actions by minimizing its own internal cost and the external cost induced by the presence of a price. 
The operator adaptively updates the price to influence the behavior of the users to achieve a socially beneficial solution. 
\edit{The adaptive updates can be treated as a negotiation process: the customer and the system operator are negotiating on a future price that optimizes the social welfare of the system.} 
Note that as long as users can optimize their action\edit{s} given a price, it suffices for the operator to just observe the actions of the users, and the operator does not need to know or learn users' private information. 

These types of incentive designs where an operator sets a payment (or sometimes called a tax) have received significant interest. \edit{Typically, the goal is to provide incentives to selfish users such that the resulting equilibrium is closer to the socially optimal solution~\cite{bacsar1984affine,ho1982control,paccagnan2019incentivizing}. In the context of electricity markets, a large number of demand response strategies have been proposed, ranging from alert/text-based signals~\cite{peplinski2023residential}, to pricing~\cite{vardakas2014survey}, to direct load control~\cite{chen2014distributed}. In this paper, we focus on developing pricing mechanisms.} The incentive mechanism we adopt is from~\cite{maheshwari2022inducing}, which is a type of dynamic incentive that evolves with the actions of the users. The key technical questions in these type\edit{s} of mechanisms are twofold: whether the price update converges, and if so, whether it leads to socially optimal user behaviors. 

Existing results in answering these two questions often require strong assumptions that do not readily apply to electricity markets. In~\cite{ratliff2020adaptive}, the cost function of the users and the system cost are assumed to be linear. The results in~\cite{Liy2021inducing} generalize the class of cost functions but assume the users\edit{'} strategies are convex in the price. This assumption often does not hold even in relatively simple games \edit{induced by the price}. \edit{The work of \cite{li2011optimal} is similar to ours but makes much stronger assumptions that the system and users’ costs are convex and smooth. Moreover, our algorithm circumvents the need for each user to determine the incremental demand at each step, which is challenging, especially for learning-based methods. } The work in~\cite{maheshwari2022inducing}, which acts as the impetus to this paper, shows that the price converges and leads to social\edit{ly} optimal behavior assuming strict convexity of the cost functions. 

In this paper, we relax the assumptions needed in~\cite{ratliff2020adaptive,Liy2021inducing,maheshwari2022inducing}, and show that the two\edit{-}time-scale algorithm leads to price dynamics that both converge and induce optimal social behavior of the users. For example, we show that in some settings, the only assumption needed is that a user's action (electricity consumption\edit{)} is decreasing in price. \edit{This assumption is intuitive because higher price naturally leads to lower consumption. }
This result would then apply to users who use learning algorithms to manage their demands. We also generalize the type of social costs that operators could consider. \edit{For example, in addition to widely used quadratic functions and time-of-use pricing, our framework provides a way to implement peak pricing through designing the social cost. }

In our proof of convergence, we construct a continuous\edit{-}time dynamical system and its associated Lyapunov functions. Then we show that the price leads to user actions that are optimal for the global optimization problem by analyzing the optimality conditions. The paper is organized as follows. In Section~\ref{sec:model}, we describe the setup of our model. In Section~\ref{sec:theory}, we provide the theoretical results, and Section~\ref{sec:simulation} illustrates these results with several simulated case studies. We conclude and give some thoughts in future directions in Section~\ref{sec:conclusion}. 

%
%

%% file: model.tex
In this section, we formally state the problem of obtaining socially optimal energy usage and describe an iterative procedure to compute the price seen by each of the individual users. 


\subsection{Global Optimization Problem}
We consider a system with $N$ users. The action of each user is their energy consumption over $T$ time periods. For example, $T=24$ if we are considering hourly loads of the users in a day. Let $\x_i\in \mathbb{R}^T$ denote the electricity consumption of user $i$.\footnote{\edit{We use bold notation to indicate vectors and matrices.}} We assume $\x_i$ takes value in a compact set $\mathcal{X}_i \subset \mathbb{R}^T$.\footnote{The constraints that a user has can be encoded in this set.} The notation $x_{i,t}$ denote the energy consumption of user $i$ at the $t$'th period. Each user is interested in optimizing its cost function, $f_i:\mathbb{R}^T \rightarrow \mathbb{R}$. For example, if $f_i(\x_i)=\frac{1}{2} (\x_i-\bar{\x}_i)\edit{^2}$, the user can be interpreted as minimizing its energy usage to some predetermined setpoint $\bar{\x}_i$. Note that in this paper we do not assume any symmetry between the users, nor do we restrict $f_i$ to be in particular parametric forms. Of course, some assumptions on $f_i$'s are required, and these will be stated closer to the technical results \edit{section}. 

We also place a cost on the system in serving the \edit{total } energy demanded by the users. We do not distinguish between which user is using the energy, and the system cost is in the form of $g\left(\sum_{i=1}^N  \x_i\right)$. Note the sum in $g$ is over the users, and $g$ is a function from $\mathbb{R}^T$ to $\mathbb{R}$. Throughout this paper, we will assume that $g$ is convex or strictly convex. Several examples of $g$ include  $g\left(\sum_i \x_i\right)=c\left(\sum_i \x_i \right)$, where $c(\cdot)$ is the cost of procuring energy; $g\left(\sum_i \x_i\right)=||\sum_i \x_i||^2$, where the system tries to reduce its total energy consumption; and $g\left(\sum_i \x_i\right)=\max \left(\sum_i x_{i,1},\dots, \sum_i x_{i,T}\right)$, which represent a cost on the peak energy demand~\cite{yan2015enabling}. 

The system operator is interested in solving a global social welfare problem:
\begin{equation}\label{eqn:social_welfare}
    \min_{\x} \quad \sum_i f_i(\x_i) + g\left(\sum_i \x_i \right).
\end{equation}
The first term represents the sum of the users' costs, while the second represents a social or system cost. This problem has been studied in many settings. When the operator has full information about $f_i$'s and $g$, and can set the demand of the users, this is essentially a variant of the economic (load) dispatch problem. This has been extensively studied and readers can refer to~\cite{pinson2014benefits,huang2019demand} and the references within. 

In practice, it is often the case that the system operator does not know the exact cost functions of the users. This could stem from privacy concerns, where the user does not communicate their cost to the operator. Another way that incomplete information can arise is that the users themselves do not have a closed-form description of their own costs. Learning algorithms are becoming increasingly popular for demand management, where a user's action is trained without explicitly learning a cost model~\cite{motokiwaterheater,wang2020deep,vazquez2019reinforcement}. In this setting, an iterative algorithm is often used, where local steps and global steps alternate and try to converge to the solution of the social welfare problem in~\eqref{eqn:social_welfare}. 

When $f_i$'s are convex (or strictly convex), differentiable and frequent communications are possible, primal/dual type of gradient algorithms can be used to solve it. More efficient algorithms exist when $f_i$'s are assumed to be in particular parametric forms (mostly quadratic~\cite{Li19,khezeli2017risk}). Drawbacks of these approaches appear when the local problem is more complicated. For example, suppose user $i$ uses tabular $Q$-learning to determine its actions. Then it is difficult to apply existing distributed optimization algorithms. 

In this paper, we study a strategy to iteratively solve \eqref{eqn:social_welfare}, where we can weaken the existing assumptions on the users. We do this through price updates and more details are given next.

\subsection{User's Optimization Problem}
From a user's perspective, it receives a price vector, $\p\in \mathbb{R}^T$ from the system operator. We will define how the operator arrives at this price in the next section, but for now, we treat it as a given vector. 
User $i$ solves the following optimization problem:
\begin{equation}\label{eqn:local_opt_prob}
    \min_{\x_i \in \mathcal{X}_i} \quad  f_i(\x_i) + \p^\T \x_i. 
\end{equation}
We denote the solution of this problem as $\xis(\p)$. We do not consider the details of how user $i$ solves \eqref{eqn:local_opt_prob} to obtain $\xis(\p)$. If $f_i$ is available, the user may explicitly solve \eqref{eqn:local_opt_prob}. In~\cite{motokiwaterheater}, a user represents a water heater and has the goal of balancing the discomfort of receiving underheated or overheated water with the cost of power. Since the discomfort is easily expressed in terms of the temperature of the water but not in terms of power, a Q-learning algorithm is used to look up $\xis(\p)$. 

In this paper, we make the following assumption: 
\begin{assumption} \label{assump:local_unique}
Given a price $\p$, the solution $\xis(\p)$ to \eqref{eqn:local_opt_prob} is unique for all users $i$. 
\end{assumption}
This says that the behavior of a user is uniquely determined by the price it sees from the system operator. 
This condition is true under a wide range of conditions, and we note it is much weaker than the ones made in existing literature. \edit{For example,  it holds if $f_i$ is strictly convex (as assumed in~\cite{maheshwari2022inducing,ratliff2020adaptive,liu2014pricing, li2011optimal}, or if the user used a deterministic algorithm regardless of convexity~(not covered by the conditions in~\cite{maheshwari2022inducing,ratliff2020adaptive,liu2014pricing,li2011optimal}).} If a probabilistic algorithm is used, the results in the paper hold with respect to the averaged solution, but the analysis becomes much more cumbersome. \edit{A more refined understanding of possible stochastic behaviors is an important future direction for us.}

\edit{It's worth mentioning that this assumption is not strictly necessary for the adaptive algorithm to find socially optimal incentives. 
However, without this assumption, the updates may oscillate between several different solutions. 
When the adaptive algorithm does not converge to a unique solution, the notion of convergence becomes more difficult to characterize. 
In response to the potential oscillations that destabilize the system, the system operator could develop a consistent tie-breaking procedure (e.g. consuming at the earliest hour) so the system would converge to a fixed point. 
In addition, we could expand the allowed behavior of the agents by
considering ideas such as convergence in expectation. 
}



\subsection{Price Update}
Now we consider how $\p$ should be updated by the system operator. Following the scheme in \cite{maheshwari2022inducing}, we define the externality $\e$ as the marginal social cost arising from the term $g$:
\begin{equation}\label{eqn:e}
    \e(\x) = \nabla_\mathbf{z} g(\mathbf{z})|_{\mathbf{z}=\sum_i \x_i}
\end{equation}
\edit{where $\x = \{\x_i\}_{i=1}^N$}. The price is updated as
\begin{equation}\label{eqn:learning_dynamics}
    \p_{k+1} = (1-\beta_k)\p_k + \beta_k \e(\x^{*}(\p_k)). 
\end{equation}
\edit{where $\x^{*}(\p_k) = \{\x^{*}_i(\p_k)\}_{i=1}^N$ and} the following conditions hold $\sum_{k=1}^\infty \beta_k = \infty$ and $\sum_{k=1}^\infty \beta_k^2 < \infty$; e.g., $\beta_k = 1 / k$. These assumptions about the step size are standard assumptions that allow us to prove the convergence of the price updates.\footnote{\edit{We note that step sizes of $1/k$ are required for theoretical reasons; however, we find that experimentally that other step sizes can lead to the same fixed point with fewer iterations.}}

The main technical results of the paper are to show that under mild assumptions, the dynamics in \eqref{eqn:learning_dynamics} converge to a unique $\p^{*}$ and $\xis(\p^{*})$ is the minimizer of \eqref{eqn:social_welfare}. This price update scheme is attractive because the operator doesn't need to know or learn about the users.\footnote{In fact, the operator does not even need to know how many users there are in the system since it simply broadcasts the price to all users.} As long as the sum of decisions $\x_i$'s are observed, a unique $\p^{*}$ can be found to induce globally optimal behavior. 

%% file: algorithm.tex
\begin{algorithm}
    \caption{Adaptive Pricing: An iterative method that solves the global optimization problem (\ref{eqn:social_welfare}).}
    \label{alg:adaptive_pricing}
    \begin{algorithmic}[1]
        \State {Initialize $\p\in \mathbb{R}^T$}
        \For{$k = 1, 2, \dots$}

            \vspace{0.1cm} \State {\textit{// Solve Local Optimization}}
            \For{$i = 1, 2, \dots, N$}
                \State {$\xis(\p) = \arg \min_{\x_i \in \mathcal{X}_i} \;  f_i(\x_i) + \p^\T \x_i$}
            \EndFor
            
            \vspace{0.1cm} \State {\textit{// Price Update}}
            \State {Choose step size: $\beta  = 1 / k$ }
            \State {Compute externality: $\mathbf{e} = \nabla g\Big(\sum_{i=1}^N \xis(\p)\Big)$}
            \State {Compute new price: $\p = (1-\beta)\p + \beta \mathbf{e}$}


            
        \EndFor
    \end{algorithmic}
\end{algorithm}

%% file: theoretical_results.tex
\label{sec:theoretical_results}

The first step to analyzing the system in~\eqref{eqn:learning_dynamics} is to approximate it as a continuous system. A standard result in dynamical systems is that as  $k \rightarrow \infty$ and the step-sizes $\beta_k$ are asymptotically going to zero, \eqref{eqn:learning_dynamics} behaves as~\eqref{eqn:continuous_dynamics} \cite{maheshwari2022inducing, vivek1997stochastic}: 
\begin{equation}\label{eqn:continuous_dynamics}
\dot{\p} = \e(\x^{*}(\p)) - \p. 
\end{equation}
We want \eqref{eqn:continuous_dynamics} to have two properties. The first is that it has a unique equilibrium, and the second is that this equilibrium induces the solutions to the local optimization problem in \eqref{eqn:local_opt_prob} that's simultaneously solving the global problem in \eqref{eqn:social_welfare}. 

The proof of these results involves constructing a Lyapunov function based on suitable assumptions on $f_i$'s and $g$. We feel it is easier to break the proof into two parts, and first consider the case of $T=1$ and then $T>1$. When $T=1$, the proof is fairly short and the intuition carries over to the $T>1$ case, where the math becomes more cumbersome. 


\subsection{Single time-period case (T = 1)}
For a single time-dimensional case, $x_i \in \mathbb{R}$, $g$ is a scalar function, and the dynamical system in \eqref{eqn:continuous_dynamics} is also a scalar system.  We have the following result
\begin{theorem} \label{thm:t_single}
    Suppose Assumption~\ref{assump:local_unique} holds and $x_i^{*}(p)$ is decreasing for each $i$ and $g$ is convex and differentiable. Then the scalar dynamical system \eqref{eqn:continuous_dynamics} is globally asymptotically stable and has a unique fixed point $p^{*}$. The solutions $x_i^{*}(p^{*})$ solve the global optimization problem in~\eqref{eqn:social_welfare}. 
\end{theorem}
The decreasing condition on $x_i^{*}(p)$ says that if the price increases, then the energy consumption of user $i$ will not increase (users will not consume more if energy becomes more expensive). This condition is satisfied for most types of cost functions. For example, if $f_i$'s are convex. The decreasing condition is also satisfied for situations where $f_i$'s are nonconvex~\cite{wan2022nonlinear}. Therefore, our pricing algorithm is applicable to a wider range of problems than existing pricing schemes that require both $f_i$'s and $g$ to be convex~\cite{maheshwari2022inducing}. 

\begin{proof}
We first show that a unique equilibrium of \eqref{eqn:continuous_dynamics} exists. When $T=1$, \eqref{eqn:e} simplifies to $e(x) = g'(\sum_{i} x_i)$. If $g$ is convex, its derivative is increasing and $e$ is an increasing function. Since $x_i^{*}(p)$ is decreasing in $p$, the function of $p$ obtained by summing over $i$,  $\sum_i x_i^{*}(p)$, is also decreasing in $p$. 
Using the fact that for scalar functions, the composition of an increasing function with a decreasing function is decreasing, we have $e(\sum_i x_{i}^{*} (p))$ being decreasing in $p$.

Looking at the right hand side of \eqref{eqn:continuous_dynamics}, $e(x^{*}(p))-p$ is a strictly decreasing function ranging from $\infty$ to $-\infty$ as $p$ goes from $-\infty$ to $\infty$. Therefore it takes the value $0$ at exactly one point, and we denote that point by $p^{*}$, and it is an equilibrium point of the dynamical system in \eqref{eqn:continuous_dynamics}. 

Next, we show that the dynamical system is globally asymptotically stable around $p^{*}$. We do this by constructing the following Lyapunov function:
\[ V(p) = -\int_{p^{*}}^p [e(x^{*}(q)) - q] \, dq.  \]
Using the fact that the integrand is strictly decreasing, is easy to verify that $V(p) > 0$ if $p \neq p^{*}$ and $V(p^{*})=0$. The time derivative of $V$ is 
\[
    \dot{V}(p)= V'(p) \dot{p} = -[e(x(p)) - p]^2,  
\]
which is negative except when $p=p^{*}$. 
Therefore, $p^{*}$ is globally asymptotically stable. 

Finally, we show that the $p^{*}$ induces globally optimal behavior. That is, the solutions of the local problems $x_i^{*}(p^{*})$ solves the global optimization problem in \eqref{eqn:social_welfare}. To show this, let $\hat{x}$ denote an optimal solution to \eqref{eqn:social_welfare}.  For notational simplicity, define $\hat{x}_s =\sum_i \hat{x}_i$ and $x_s^{*}=\sum_i x_i^{*}(p^{*})$. Suppose $\hat{x} \neq x^{*}(p^{*})$, and $\hat{x}$ achieves a lower social welfare cost. Then we have
\begin{equation}\label{eqn:hat_less_star}
\sum_i f_i(\hat{x}_i)+g(\hat{x}_s) < \sum_i f_i ( x_i^{*}(p^{*}))+g(x_s^{*}).
\end{equation} 
But $x_i^{*}(p^{*})$ are the minimizers of the local optimization problem in~\eqref{eqn:local_opt_prob}, we have
\begin{equation} \label{eqn:star_less_hat}
f_i ( x_i^{*}(p^{*}))+g'(x_s^{*}) x_i^{*}(p^{*}) < f_i ( \hat{x}_i)+g'(x_s^{*}) \hat{x}_i, \; \forall i, 
\end{equation}
where we use the fact that $p^{*}=g'(x_s^{*})$ and the inequality is strict because of the uniqueness assumption in Assumption~\ref{assump:local_unique}. Adding \eqref{eqn:hat_less_star} and \eqref{eqn:star_less_hat} and simplifying the expressions, we get 
\[ g(\hat{x}_s) < g(x_s^{*}) + g'(x_s^{*}) (\hat{x}_s-x_s^{*}).\]
This is the first-order condition of a strictly \emph{concave} function, and it contradicts the fact that $g$ is \emph{convex}. Hence the assumption that there exists a more optimal solution to \eqref{eqn:social_welfare} than $x^{*}$ is not possible. 
\end{proof}

\subsection{Multi-Time Period Case (T > 1)}
The analysis in the single time period case relies on the monotonicity of several functions. To extend that to the multi-time period case, we extend the notion of monotonicity to vector-valued functions:
\begin{definition} \label{def:monotone}
Consider a vector-valued function $h: \mathbb{R}^n \rightarrow \mathbb{R}^n$. We say that $h$ is increasing if $[h(\x)-h(\x')]^\T (\x-\x')\geq 0$. It is strictly increasing if the inequality is strict when $\x \neq \x'$. A function is (strictly) decreasing if its negation is (strictly) increasing. 
\end{definition}
We adopt Definition~\ref{def:monotone} since it is a natural generalization of monotonicity for a scalar function in the following sense: the gradient of a differentiable convex function is increasing.\footnote{We note there exists other definition of monotonicity for vector-valued functions. For example, sometimes a function is said to be monotone if it is monotone in every coordinate. This definition, however, is too restrictive for our purposes.}  

In the $T=1$ setting, the proof for Theorem~\ref{thm:t_single} has three parts. First, we showed that a unique equilibrium price exists, then we show that the dynamical system in \eqref{eqn:continuous_dynamics} is asymptotically stable around this price, and finally that this price induces socially optimal behavior. In the $T>1$ setting, the first and last parts generalize directly from the $T=1$ setting with minimal changes. However, asymptotic stability is considerably more difficult since the composition of two monotone functions is no longer monotone.\footnote{This is not even true for linear functions in the vector-valued case.} Therefore, we first state a lemma about the uniqueness and social optimality of the equilibrium price. Then we provide two conditions where the iterative algorithm converges asymptotically to this equilibrium. 

\begin{lemma} \label{lem:t_multiple}
    Consider the dynamical system in \eqref{eqn:continuous_dynamics}. If $\xis(\p)$ is a decreasing function of $\p$, then there exists a unique equilibrium $\p^{*}$. Furthermore, under Assumption~\ref{assump:local_unique}, $\x^{*}(\p^{*})$ is the optimal solution to social welfare problem in~\eqref{eqn:social_welfare}. 
\end{lemma}
\begin{proof}
    The proof of these two results is analogous to the $T=1$ case. By sweeping $\p$ from $\infty$ to $-\infty$ and using the fact that $\e(\x^{*}(\p))-\p$ is strictly decreasing, we can conclude that there exists a unique $\p^{*}$ such that  $\e(\x^{*}(\p^{*}))-\p^{*}=0$. 
    For the second part of the lemma, it suffices to change a derivative to a gradient in \eqref{eqn:star_less_hat} and the product between scalars to a dot product between vectors. All other steps remain the same.  
\end{proof}

Next, we state two conditions where the dynamical system in~\eqref{eqn:continuous_dynamics} converges: 
\begin{theorem} \label{thm:t_multi}
Suppose the local solutions $\xis(\p)$ are decreasing in $\p$. 
The dynamical system in~\eqref{eqn:continuous_dynamics} is asymptotically stable if one of the following conditions are true
\begin{enumerate}
    \item $g(\mathbf{z})=\frac{1}{2} \mathbf{z}^\T \mathbf{B} \mathbf{z}$ for some positive definite $\mathbf{B}$. 
    \item Each $f_i$ is twice differentiable and strictly convex and $g$ is strictly convex and differentiable.  
\end{enumerate}
\end{theorem}
Before presenting the proof, we comment on why these settings are potentially interesting and practical. 
A quadratic cost is the most common form used in the literature~\cite{kirschen2018fundamentals} since it includes standard pricing schemes such as flat rates and time-of-use (ToU) rates. The assumption on each of the users is minimal since we just require that an increase in energy price leads to a decrease in energy consumption. 

The second setting generalizes the price function by assuming a more strict form of the local costs $f_i$'s. The reason this result is useful is that it includes peak pricing, which is being adopted by a large number of users \edit{\cite{DOE_peak,Smart_peak,newsham2010effect,wang2015time}}. More specifically, peak pricing has the form of $g(\mathbf{z})=c \cdot \max(z_1,\dots,z_T)$, where $c$ is some constant. 
This function is convex, but not strictly convex nor differentiable. But it can be approximated by the so-called LogSumExp function, defined as 
\begin{equation}\label{eqn:LSE}
\LSE(z_1,\dots,z_T)= \frac{1}{\alpha} \log (\exp(\alpha z_1)+\dots+\exp(\alpha z_T)).
\end{equation}
The LSE function is strictly convex if at least one of its arguments is positive and approximates the max function arbitrarily as $\alpha$ grows~\cite{nielsen2019monte}. 

\begin{proof}
    We prove Theorem~\ref{thm:t_multi} by constructing appropriate Lyapunov functions. In the first case, we use a quadratic Lyapunov function
    \begin{equation} \label{eqn:lya_quad}
        V(\p)=\frac{1}{2} (\p -\p^{*})^\T \mathbf{B}^{-1} (\p-\p^{*}).  
    \end{equation}
    It is clear that $V(\p)=0$ if $\p=\p^{*}$ and $V(\p) >0$ otherwise. 

    In the second setting, we use a different Lyapunov function, defined as 
    \begin{equation} \label{eqn:lya_bregman}
    \begin{split}
       V(p) & = g\left(\sum_{i} \x_i^{*}(\p)\right) - g\left(\sum_{i} \x_i^{*}(\p^{*})\right) \\
     &  - \nabla g\left(\sum_{i} \x_i^{*}(\p^{*})\right)^{\T}\left(\sum_{i} (\x_i^{*}(\p)-\x_i^{*}(\p^{*}))\right) \\
    &  - \sum_{i} (\x_i^{*} (\p)- \x_i^{*}(\p^{*}))^{\T}(\p-\p^{*}). 
    \end{split}
    \end{equation}
    This function is based on the Bregman's divergence of $g$, and using the fact that $g$ is strictly convex, $V(\p)=0$ if $\p=\p^{*}$ and $V(\p) >0$ otherwise.

    In Appendix~\ref{app:t_multi}, we show that the time derivative of both Lyapunov functions are negative when $\p \neq \p^{*}$.
\end{proof}
The proof above illustrates the difficulty of showing a general result when $T>1$ since we need two different Lyapunov functions. However, we believe this is a consequence of our proof technique, and we conjecture that the twice differentiable and strictly convex condition in Theorem~\ref{thm:t_multi} can be relaxed to be just $\xis(\p)$ is strictly decreasing for all $i$.

%% file: simulation_results.tex
In this section, we conduct numerical studies that illustrate the theoretical results in this paper, in particular the convergence properties of the update in~\eqref{eqn:learning_dynamics}. Our code for the numerical simulations can be found at \url{https://github.com/socially-optimal-energy}.

\begin{figure}[ht] 
    \centering
    \includegraphics[width=\columnwidth]{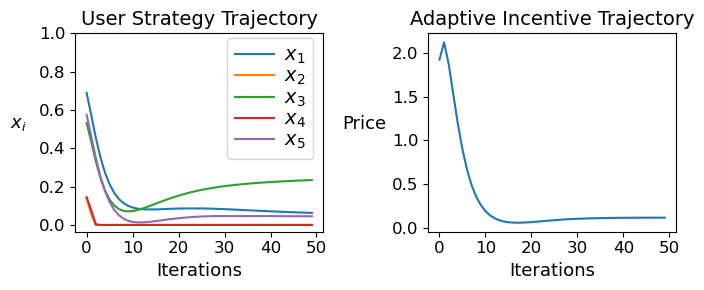}
    \caption{Convergence of user actions and price incentive for a system with 5 users and a single time-period. Both the actions and the price converge quickly.}
    \label{fig:N5T1}
\end{figure}


\subsection{Single Time-Period}
The first is a simple single time-period ($T=1$) example that clearly shows the convergence to the global solution (see Fig.~\ref{fig:N5T1}). Suppose there are 5 users ($N=5$) that have actions $x_i \in \mathbb{R}$ which can automatically adapt to changes in price.\footnote{\edit{In this example, we assume that the users and utility are in a negotiating a future price (e.g., a day-ahead price) and the algorithm converges by the time the price needs to be set and the load is realized.}} The user cost functions are $f_i(x_i)=\frac{1}{2}(x_i-\bar{x}_i)^2$ and the global cost is $g(\sum_i x_i)=(\sum_i x_i)^2$. Both actions and the price converge quickly, and it is easy to check that they converge to the optimal values (which can be computed in this simple setup). 

\subsection{Peak Pricing of Multiple Time-Periods}
Here, we consider multiple periods, with $T=24$. Then user~$i$'s action is $\x_i \in \mathbb{R}^{24}$. For simplicity, let $f_i(\x_i)=\frac{1}{2}(\x_i-\bar{\x}_i)^2$ for some given $\bar{\x}_i$'s. Suppose the system adopts peak pricing, approximated as $g(\sum_i \x_i)=\lambda\LSE(\sum_i \x_i) $ for some $\lambda$, where the LSE function was defined in~\eqref{eqn:LSE}.  
Following~\eqref{eqn:e}, the externality of each user is calculated as: 
\begin{equation*}
    \e(\x) = \nabla_\mathbf{z} g(\mathbf{z})|_{\mathbf{z}=\sum_i \x_i} = \frac{1}{ \sum_{t=1}^T \exp{(\alpha \mathbf{z}_t})}\exp{(\alpha \mathbf{z})}.
\end{equation*}
For 10 users, the convergence of price and total demand is shown in Fig.~\ref{fig:N10T24}, where both converge quickly to their final values. The comparison of initial and final price and load profiles are shown in Fig.~\ref{fig:profiles}. The initial prices (chosen uniformly at random) induce a total (summed over the users at each time period) load profile that is uneven. After convergence, load profiles become much more even, which is what we expect to see when the peak system load is minimized. Note the price that achieves this still shows variations across time periods. 

\begin{figure}[ht] 
    \centering
    \includegraphics[width=\columnwidth]{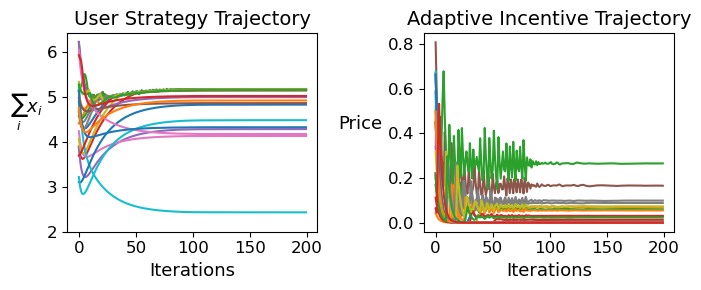}
    \caption{The top figure shows the convergence in the sum of the users{'} actions at each time period. The bottom figure shows the convergence of price.}
    \label{fig:N10T24}
\end{figure}


\begin{figure}[ht] 
    \centering
    \includegraphics[width=0.95\columnwidth]{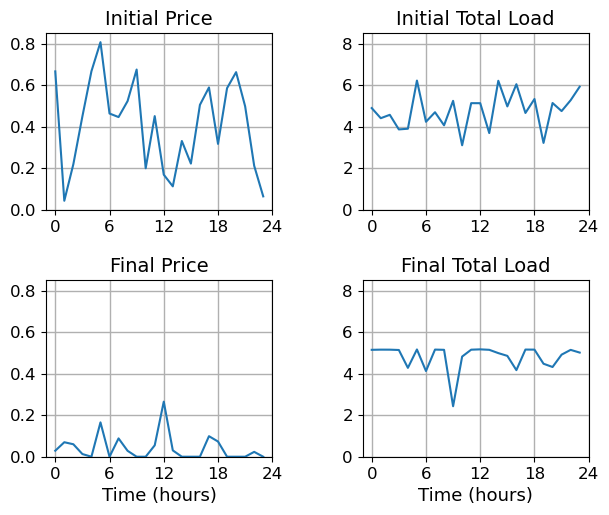}
    \caption{The initial and converged price and demand profiles demonstrate that this adaptive pricing framework is effective in reducing system peak. }
    \label{fig:profiles}
\end{figure}


\subsection{Water Heater Load Optimization with Q-Learning}
Now we consider a case where users{'} actions are determined by a learning algorithm and show that the $f_i$'s need not be differentiable or convex. We consider a group of water heaters where $\x_i \in \mathbb{R}^{96}$ represents its daily load profile measured every 15-minutes, $f_i$ represents user discomfort, and $g$ is the $\LSE$ function. \edit{In this example, users can flexibly shift their energy consumption earlier to avoid high prices and store the energy (hot water) until it is needed.} We model and optimize each user following the learning algorithm in~\cite{motokiwaterheater} and numerically check that $\x_i^{*}$ is decreasing in $\p$.

For simplicity, we assume that \edit{each user's demand for hot water is a binary vector drawn from a sequence of Bernoulli random variables.} We assume the demand and all system parameters are known and so given a price $\p$, we can solve for $\x^{*}(\p)$ exactly using dynamic programming. Fig.~\ref{fig:water_heater_simulation}, shows the initial and final price and the corresponding distributions of total energy consumption of all the users. Initially, the system cost is 12.781 and the user cost is 0.630. At convergence, the system cost decreases to 10.397 and the user cost to 0.627. Overall, the adaptive pricing procedure reduces the social welfare cost by 2.387 (17.8\%).


\begin{figure}[ht] 
    \centering
    \includegraphics[width=0.95\columnwidth]{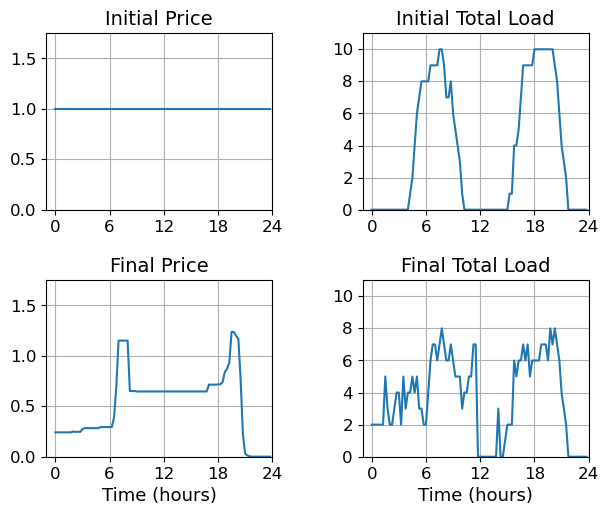}
    \caption{Simulation results for water heater optimization. There \edit{are} 10 users minimizing discomfort using a Q-learning algorithm~\cite{motokiwaterheater}, and the system operator tries to minimize the peak load. The initial price is chosen to be flat, which leads to a profile with high peaks. After several iterations, the price becomes uneven, and the final load has much lower peaks. }
    \label{fig:water_heater_simulation}
\end{figure}


%% file: appendix.tex
Consider the first condition in Theorem~\ref{thm:t_multi}. Again we let $\x_s$ denote the sum $\sum_i \x_i$. The time derivative of the Lyapunov function in \eqref{eqn:lya_quad} is
\begin{align*}
    \dot{V}(\p) & = \nabla V(\p)^\T \dot{\p} \\
    &= (\p-\p^{*})^\T \mathbf{B}^{-1} (\e(\x_s^{*}(\p))-\p) \\
    &\stackrel{(a)}{=} (\p-\p^{*})^\T \mathbf{B}^{-1} (\e(\x_x^{*}(\p))-\p-\e(\x_s^{*}(\p^{*}))+\p^{*}) \\
    &\stackrel{(b)}{=} (\p-\p^{*})^\T \mathbf{B}^{-1} (\mathbf{B} (\x_s^{*}(\p)-\x_s^{*}(\p^{*}))-\p+\p^{*}) \\
    &=  (\p-\p^{*})^\T (\x_s^{*}(\p)-\x_s^{*}(\p^{*})) \\
    & \;\;\;\; - (\p-\p^{*})^\T \mathbf{B}^{-1} (\p-\p^{*}) \\
    &\stackrel{(c)}{\leq} 0,
\end{align*}
where $(a)$ follows from the fact that $\p^{*}$ is the equilibrium; $(b)$ follows from the definition $\e=\nabla g$; $(c)$ follows from the assumption that each $\x_i^{*}(\p^{*})$ is decreasing (and hence their sum is) and $\mathbf{B}$ is positive definite. Furthermore, $V(\p)=0$ only if $\p=\p^{*}$.

Next, we consider the second condition in Theorem~\ref{thm:t_multi}. The time derivative of the Lyapunov function in \eqref{eqn:lya_bregman} is
\begin{align}
    \dot{V}(\p) & = \nabla V(\p)^\T \dot{\p} \nonumber\\
    &= \nabla V(\p)^\T (\e(\x_s^{*}(\p))-\e(\x_s^{*}(\p^{*}))-\p+\p^{*}) \nonumber \\
    &  = (\e(\x_s^{*}(\p))-\e(\x_s^{*}(\p^{*}))-\p+\p^{*})^\T \left[\nabla_{\p} \x_s^{*} (\p) \right] \nonumber \\
     & \;\;\;\; (\e(\x_s^{*}(\p))-\e(\x_s^{*}(\p^{*}))-\p+\p^{*}) \label{eqn:g_x} \\
    & - (\x_s^{*}(\p)-\x_s^{*}(\p^{*}))^\T (\e(\x_s^{*}(\p))-\e(\x_s^{*}(\p^{*}))) \nonumber \\
    & + (\x_s^{*}(\p)-\x_s^{*}(\p^{*}))^\T(\p-\p^{*}). \nonumber 
\end{align}
It's easy to see that $\dot{V}(\p)=0$ when $\p=\p^{*}$. We first look at the last two terms. Since $\e=\nabla g$ and $g$ is strictly convex, $\e$ is strictly increasing, and $(\x_s^{*}(\p)-\x_s^{*}(\p^{*}))^\T (\e(\x_s^{*}(\p))-\e(\x_s^{*}(\p^{*})))>0$ if $\p \neq \p^{*}$. By assumption, $\x_s^{*}(\p)$ is decreasing and $(\x_s^{*}(\p)-\x_s^{*}(\p^{*}))^\T(\p-\p^{*}) \leq 0$. 

Next, consider the term in \eqref{eqn:g_x}. The object $\nabla_{\p} \x_s^{*} (\p)$ is the gradient of the vector $\x_s^{*}(\p)$, hence it is a $T \times T$ matrix and it suffices to show $\nabla_{\p} \x_s^{*} (\p)$ is negative definite. Using the fact that summations of negative definite matrices are negative definite, it is enough to show that $\nabla_{\p} \x_i^{*} (\p)$ is negative definite for all $i$. Recall $\x_i^{*}(\p)$ is the optimal solution to the local optimization problem, and it satisfies the first\edit{-}order condition 
\[ \nabla_{\x_i} f_i (\x_i)+\p=0.\]
Taking \edit{the} derivative of $\p$ on both sides, applying the chain rule, and rearranging leads to:
\[ \mathbf{I} = - \left[ \nabla_{\p} \x_i^{*}(\p)\right] \left[\mathbf{H}_{\x_i}(\x_i^{*}(\p))\right]  \]
where $\left[\mathbf{H}_{\x_i}(\x_i^{*}(\p))\right]$ is the Hessian of $f_i$ evaluated at $\x_i^{*}(\p)$. Since $f_i$ are assumed to be strictly convex and twice differentiable, $\left[\mathbf{H}_{\x_i}(\x_i^{*}(\p))\right]$ is positive definite. Rearranging the above equation gives $ \nabla_{\p} \x_i^{*}(\p) =- \left[\mathbf{H}_{\x_i}(\x_i^{*}(\p))\right]^{-1}.$
Therefore $\nabla_{\p} \x_i^{*}(\p)$ is negative definite and  $\dot{V}(\p)<0$ when $\p \neq \p^{*}$. 

%% file: main.bbl
\begin{thebibliography}{10}
\providecommand{\url}[1]{#1}
\csname url@samestyle\endcsname
\providecommand{\newblock}{\relax}
\providecommand{\bibinfo}[2]{#2}
\providecommand{\BIBentrySTDinterwordspacing}{\spaceskip=0pt\relax}
\providecommand{\BIBentryALTinterwordstretchfactor}{4}
\providecommand{\BIBentryALTinterwordspacing}{\spaceskip=\fontdimen2\font plus
\BIBentryALTinterwordstretchfactor\fontdimen3\font minus
  \fontdimen4\font\relax}
\providecommand{\BIBforeignlanguage}[2]{{%
\expandafter\ifx\csname l@#1\endcsname\relax
\typeout{** WARNING: IEEEtran.bst: No hyphenation pattern has been}%
\typeout{** loaded for the language `#1'. Using the pattern for}%
\typeout{** the default language instead.}%
\else
\language=\csname l@#1\endcsname
\fi
#2}}
\providecommand{\BIBdecl}{\relax}
\BIBdecl

\bibitem{pinson2014benefits}
P.~Pinson, H.~Madsen \emph{et~al.}, ``Benefits and challenges of electrical
  demand response: A critical review,'' \emph{Renewable and Sustainable Energy
  Reviews}, vol.~39, pp. 686--699, 2014.

\bibitem{zhang2017robust}
C.~Zhang, Y.~Xu, Z.~Y. Dong, and K.~P. Wong, ``Robust coordination of
  distributed generation and price-based demand response in microgrids,''
  \emph{IEEE Transactions on Smart Grid}, vol.~9, no.~5, pp. 4236--4247, 2017.

\bibitem{tindemans2015decentralized}
S.~H. Tindemans, V.~Trovato, and G.~Strbac, ``Decentralized control of
  thermostatic loads for flexible demand response,'' \emph{IEEE Transactions on
  Control Systems Technology}, vol.~23, no.~5, pp. 1685--1700, 2015.

\bibitem{sarker2015optimal}
M.~R. Sarker, Y.~Dvorkin, and M.~A. Ortega-Vazquez, ``Optimal participation of
  an electric vehicle aggregator in day-ahead energy and reserve markets,''
  \emph{IEEE transactions on power systems}, vol.~31, no.~5, pp. 3506--3515,
  2015.

\bibitem{o2010residential}
D.~O'Neill, M.~Levorato, A.~Goldsmith, and U.~Mitra, ``Residential demand
  response using reinforcement learning,'' in \emph{2010 First IEEE
  international conference on smart grid communications}.\hskip 1em plus 0.5em
  minus 0.4em\relax IEEE, 2010, pp. 409--414.

\bibitem{li2019constrained}
H.~Li, Z.~Wan, and H.~He, ``Constrained ev charging scheduling based on safe
  deep reinforcement learning,'' \emph{IEEE Transactions on Smart Grid},
  vol.~11, no.~3, pp. 2427--2439, 2019.

\bibitem{vazquez2019reinforcement}
J.~R. V{\'a}zquez-Canteli and Z.~Nagy, ``Reinforcement learning for demand
  response: A review of algorithms and modeling techniques,'' \emph{Applied
  energy}, vol. 235, pp. 1072--1089, 2019.

\bibitem{wang2020deep}
B.~Wang, Y.~Li, W.~Ming, and S.~Wang, ``Deep reinforcement learning method for
  demand response management of interruptible load,'' \emph{IEEE Transactions
  on Smart Grid}, vol.~11, no.~4, pp. 3146--3155, 2020.

\bibitem{shi2020multi}
Y.~Shi and B.~Zhang, ``Multi-agent reinforcement learning in cournot games,''
  in \emph{2020 59th IEEE Conference on Decision and Control (CDC)}.\hskip 1em
  plus 0.5em minus 0.4em\relax IEEE, 2020, pp. 3561--3566.

\bibitem{Li19}
P.~Li, H.~Wang, and B.~Zhang, ``A distributed online pricing strategy for
  demand response programs,'' \emph{IEEE Transactions on Smart Grid}, vol.~10,
  no.~1, pp. 350--360, 2019.

\bibitem{khezeli2017risk}
K.~Khezeli and E.~Bitar, ``Risk-sensitive learning and pricing for demand
  response,'' \emph{IEEE Transactions on Smart Grid}, vol.~9, no.~6, pp.
  6000--6007, 2017.

\bibitem{zheng2020incentive}
S.~Zheng, Y.~Sun, B.~Li, B.~Qi, K.~Shi, Y.~Li, and X.~Tu, ``Incentive-based
  integrated demand response for multiple energy carriers considering
  behavioral coupling effect of consumers,'' \emph{IEEE Transactions on Smart
  Grid}, vol.~11, no.~4, pp. 3231--3245, 2020.

\bibitem{bacsar1984affine}
T.~Ba{\c{s}}ar, ``Affine incentive schemes for stochastic systems with dynamic
  information,'' \emph{SIAM Journal on Control and Optimization}, vol.~22,
  no.~2, pp. 199--210, 1984.

\bibitem{ho1982control}
Y.-C. Ho, P.~B. Luh, and G.~J. Olsder, ``A control-theoretic view on
  incentives,'' \emph{Automatica}, vol.~18, no.~2, pp. 167--179, 1982.

\bibitem{paccagnan2019incentivizing}
D.~Paccagnan, R.~Chandan, B.~L. Ferguson, and J.~R. Marden, ``Incentivizing
  efficient use of shared infrastructure: Optimal tolls in congestion games,''
  \emph{arXiv preprint arXiv:1911.09806}, 2019.

\bibitem{peplinski2023residential}
M.~Peplinski and K.~T. Sanders, ``Residential electricity demand on caiso flex
  alert days: a case study of voluntary emergency demand response programs,''
  \emph{Environmental Research: Energy}, vol.~1, no.~1, 2023.

\bibitem{vardakas2014survey}
J.~S. Vardakas, N.~Zorba, and C.~V. Verikoukis, ``A survey on demand response
  programs in smart grids: Pricing methods and optimization algorithms,''
  \emph{IEEE Communications Surveys \& Tutorials}, vol.~17, no.~1, pp.
  152--178, 2014.

\bibitem{chen2014distributed}
C.~Chen, J.~Wang, and S.~Kishore, ``A distributed direct load control approach
  for large-scale residential demand response,'' \emph{IEEE Transactions on
  Power Systems}, vol.~29, no.~5, pp. 2219--2228, 2014.

\bibitem{maheshwari2022inducing}
C.~Maheshwari, K.~Kulkarni, M.~Wu, and S.~S. Sastry, ``Inducing social
  optimality in games via adaptive incentive design,'' in \emph{Conference on
  Decision and Control (CDC)}.\hskip 1em plus 0.5em minus 0.4em\relax IEEE,
  2022, pp. 2864--2869.

\bibitem{ratliff2020adaptive}
L.~J. Ratliff and T.~Fiez, ``Adaptive incentive design,'' \emph{IEEE
  Transactions on Automatic Control}, vol.~66, no.~8, pp. 3871--3878, 2020.

\bibitem{Liy2021inducing}
B.~Liu, J.~Li, Z.~Yang, H.-T. Wai, M.~Hong, Y.~M. Nie, and Z.~Wang, ``Inducing
  equilibria via incentives: Simultaneous design-and-play finds global
  optima,'' \emph{arXiv:2110.01212}, 2021.

\bibitem{li2011optimal}
N.~Li, L.~Chen, and S.~H. Low, ``Optimal demand response based on utility
  maximization in power networks,'' in \emph{2011 IEEE power and energy society
  general meeting}.\hskip 1em plus 0.5em minus 0.4em\relax IEEE, 2011, pp.
  1--8.

\bibitem{yan2015enabling}
X.~Yan, D.~Wright, S.~Kumar, G.~Lee, and Y.~Ozturk, ``Enabling consumer
  behavior modification through real time energy pricing,'' in \emph{2015 IEEE
  International Conference on Pervasive Computing and Communication Workshops
  (PerCom Workshops)}.\hskip 1em plus 0.5em minus 0.4em\relax IEEE, 2015, pp.
  311--316.

\bibitem{huang2019demand}
W.~Huang, N.~Zhang, C.~Kang, M.~Li, and M.~Huo, ``From demand response to
  integrated demand response: Review and prospect of research and
  application,'' \emph{Protection and Control of Modern Power Systems}, vol.~4,
  pp. 1--13, 2019.

\bibitem{motokiwaterheater}
M.~Motoki, M.~Umeda, M.~Fripp, and A.~Kuh, ``Approximate dynamic programming
  for control of a residential water heater,'' in \emph{2015 International
  Joint Conference on Neural Networks (IJCNN)}, 2015, pp. 1--8.

\bibitem{liu2014pricing}
Z.~Liu, I.~Liu, S.~Low, and A.~Wierman, ``Pricing data center demand
  response,'' \emph{ACM SIGMETRICS Performance Evaluation Review}, vol.~42,
  no.~1, pp. 111--123, 2014.

\bibitem{vivek1997stochastic}
\BIBentryALTinterwordspacing
V.~S. Borkar, ``Stochastic approximation with two time scales,'' \emph{Systems
  \& Control Letters}, vol.~29, no.~5, pp. 291--294, 1997. [Online]. Available:
  \url{https://www.sciencedirect.com/science/article/pii/S0167691197900153}
\BIBentrySTDinterwordspacing

\bibitem{wan2022nonlinear}
Y.~Wan, T.~Kober, and M.~Densing, ``Nonlinear inverse demand curves in
  electricity market modeling,'' \emph{Energy Economics}, vol. 107, 2022.

\bibitem{kirschen2018fundamentals}
D.~S. Kirschen and G.~Strbac, \emph{Fundamentals of power system
  economics}.\hskip 1em plus 0.5em minus 0.4em\relax John Wiley \& Sons, 2018.

\bibitem{DOE_peak}
{Federal Energy Management Program}, ``Demand response and time-variable
  pricing programs: Western states,'' U.S. Department of Energy, Tech. Rep.,
  2024.

\bibitem{Smart_peak}
{SmartGrid.gov}, ``Recovery act: Time based rate programs,'' U.S. Department of
  Energy, Tech. Rep., 2024.

\bibitem{newsham2010effect}
G.~R. Newsham and B.~G. Bowker, ``The effect of utility time-varying pricing
  and load control strategies on residential summer peak electricity use: A
  review,'' \emph{Energy policy}, vol.~38, no.~7, pp. 3289--3296, 2010.

\bibitem{wang2015time}
Y.~Wang and L.~Li, ``Time-of-use electricity pricing for industrial customers:
  A survey of us utilities,'' \emph{Applied Energy}, vol. 149, pp. 89--103,
  2015.

\bibitem{nielsen2019monte}
F.~Nielsen and G.~Hadjeres, ``Monte carlo information-geometric structures,''
  \emph{Geometric Structures of Information}, pp. 69--103, 2019.

\end{thebibliography}
